\documentclass[a4paper,USenglish,cleveref,autoref, thm-restate]{lipics-v2021}
\pdfoutput=1

\usepackage{color}
\usepackage{cite}
\usepackage{amsthm}
\usepackage{amsmath}
\usepackage{amssymb}
\usepackage{mdframed}
\usepackage{braket}
\usepackage{graphicx}
\usepackage{mathtools}
\usepackage{kbordermatrix}
\usepackage{dsfont}
\usepackage[ruled,vlined]{algorithm}
\usepackage{float}
\usepackage{tablefootnote} 
\makeatletter 
\AfterEndEnvironment{mdframed}{%
 \tfn@tablefootnoteprintout%
 \gdef\tfn@fnt{0}%
}
\makeatother 

\DeclareMathOperator{\Tr}{Tr}

\usepackage{xcolor}
\definecolor{darkred}  {rgb}{0.5,0,0}
\definecolor{darkblue} {rgb}{0,0,0.5}
\definecolor{darkgreen}{rgb}{0,0.5,0}
\hypersetup{
  urlcolor   = blue,         
  linkcolor  = darkblue,     
  citecolor  = darkgreen,    
  filecolor  = darkred       
}

\newtheorem{problem}[theorem]{Problem}

\newcommand{\be}{\begin{eqnarray}}
\newcommand{\ee}{\end{eqnarray}}

\newcommand{\semigeq}{\succeq}

\newcommand{\I}{\mathbb{I}}
\newcommand{\X}{X}
\newcommand{\Y}{Y}
\newcommand{\Z}{Z}

\newcommand{\knote}[1]{\textcolor{red}{({\bf Kevin:} #1)}}
\newcommand{\onote}[1]{\textcolor{blue}{({\bf Ojas:} #1)}}

\renewcommand{\knote}[1]{}
\renewcommand{\onote}[1]{}
\nolinenumbers

\title{Application of the Level-$2$ Quantum Lasserre Hierarchy in Quantum Approximation Algorithms}
\titlerunning{Quantum Approximation from Level-$2$ Lasserre}

\author{Ojas Parekh}
{Sandia National Laboratories, Albuquerque, NM, USA\\email:  odparek@sandia.gov}
{}{}{}

\author{Kevin Thompson}
{Sandia National Laboratories, Albuquerque, NM, USA\\email:  kevthom@sandia.gov}
{}{}{}

\authorrunning{O.\, Parekh and K.\, Thompson} 



\acknowledgements{}

\relatedversion{} 

\supplement{}

\funding{Sandia National Laboratories is a multimission laboratory managed and operated by National Technology and Engineering Solutions of Sandia, LLC., a wholly owned subsidiary of Honeywell International, Inc., for the U.S. Department of Energy’s National Nuclear Security Administration under contract DE-NA-0003525.  This work was supported by the U.S. Department of Energy, Office of Science, Office of Advanced Scientific Computing Research, Accelerated Research in Quantum Computing and Quantum Algorithms Teams programs.}

\hideLIPIcs  

\EventEditors{}
\EventNoEds{2}
\EventLongTitle{International Colloquium on Automata, Languages and Programming (ICALP 2021)}
\EventShortTitle{ICALP 2021}
\EventAcronym{ICALP}
\EventYear{2021}
\EventDate{July 12--16, 2021}
\EventLocation{Glasgow, Scotland}
\EventLogo{}
\SeriesVolume{}
\ArticleNo{}

\begin{document}

\date{}
\maketitle

\begin{abstract}
    The Lasserre Hierarchy, \cite{L01, La01}, is a set of semidefinite programs which yield increasingly tight bounds on optimal solutions to many NP-hard optimization problems.  The hierarchy is parameterized by levels, with a higher level corresponding to a more accurate relaxation.  High level programs have proven to be invaluable components of approximation algorithms for many NP-hard optimization problems \cite{C08, B12, R12}.  There is a natural analogous quantum hierarchy\cite{D08, P10, B16_v2}, which is also parameterized by level and provides a relaxation of many (QMA-hard) quantum problems of interest \cite{B16_v2, B19, G19}.  In contrast to the classical case, however, there is only one approximation algorithm which makes use of higher levels of the hierarchy \cite{B16_v2}.  Here we provide the first ever use of the level-$2$ hierarchy in an approximation algorithm for a particular QMA-complete problem, so-called Quantum Max Cut \cite{G19, A20}.  We obtain modest improvements on state-of-the-art approximation factors for this problem, as well as demonstrate that the level-$2$ hierarchy satisfies many physically-motivated constraints that the level-$1$ does not satisfy.  Indeed, this observation is at the heart of our analysis and indicates that higher levels of the quantum Lasserre Hierarchy may be very useful tools in the design of approximation algorithms for QMA-complete problems.
    
\end{abstract}

\knote{Double check weights are $\geq 0$ wherever used}

\knote{Is it confusing to have $v_e$ for edge value when a very negative edge value correspond to a ``large'' edge?  I suggest swapping $x_e$ and $v_e$}
\onote{I'm fine either way, but I do see your point and am cool if you want to swap them -- we just have to do a careful readover afterwards.}

\knote{Need comment somehwere addressing why we put max mixed states on the vertices not included in the matching.  }

\section{Introduction}

The study of many body quantum systems, and their corresponding spectra is of utmost importance in many sub-fields of physics \cite{Bo12}.  These systems generally have an exponentially large dimension, so a direct calculation is intractable.  Indeed, determining the highest or lowest energy of a quantum state is the canonical QMA-hard problem \cite{Bo12, K02}, so we should not expect to solve the problem even with access to a quantum computer.  Hence, the study of algorithms which produce approximate solutions emerges as an interesting direction of study.  These problems are made even more interesting by the fact that, in contrast to the classical case \cite{V13, W11}, there are relatively few known rigorous approximation algorithms known.

\textbf{Prior work.}  The $2$-Local Hamiltonian problem has been a cornerstone of quantum complexity theory; however, it has been recently studied in the context of approximation algorithms \cite{B16_v2, B19, G19, H20, A20, P20, AGM21}.  Many of these algorithms draw inspiration from the seminal Goemans-Williamson Max Cut approximation algorithm \cite{G95} or other appropriate classical counterparts~\cite{R12}.  For classical approximation algorithms, an effective meta-algorithm is to solve a linear or semidefinite program (SDP) which relaxes the (NP-hard) optimization problem, followed by a rounding procedure which seeks to turn the optimal SDP variable into a solution in the appropriate domain (binary, integral, etc.).  The SDP provides a polynomial-time-computable bound on the optimization problem hence bounding the loss in objective allows one to bound the ratio of the objective obtained to the optimal solution (this quantity is called the approximation factor).  In the quantum case, the SDP variable is polynomial size, and the goal is to produce a (classical description) of an exponentially large quantum state, again with quantifiable loss.  Most such results use the same quantum generalization \cite{B16_v2, B19, G19} of a semidefinite programming hierarchy discovered independently by several authors in the classical case~\cite{L01, P00, G01}.  Variations in the aforementioned results~\cite{B19, G19, H20, P20} derive from differences in either the SDP used to relax the problem~\cite{H20}, changing the rounding algorithm~\cite{P20, H20, B19, G19}, or in some cases by slightly modifying the approximation algorithm and providing a better analysis for the formal proof of the approximation factor~\cite{P20}.

With only one exception, \cite{A20}, these results all have a rounding step which produces a product state.  Since there are upper bounds on the performance of product states \cite{G19}, these results all have necessarily limited performance, and it is desireable to produce non-product states for a better objective.  Another common thread in many of these works is the use of the level-$1$ instance of the quantum Lasserre Hierarchy.  As we will demonstrate in \Cref{sec:l2_is_phys}, this is a relatively loose relaxation which does not satisfy important physical constraints that a consistent quantum state would satisfy.  Hence, to get a better objective it is important to use a higher level of Lasserre, for a tighter bound on the optimal quantum state.  

There are two works of particular interest in the current context: \cite{A20} and \cite{B16_v2}, which we comment on.  We will first need to formally describe a specific $2$-Local Hamiltonian problem, introduced as a quantum analog of Max Cut~\cite{G19}.  Note that here and throughout the paper, we will use the notation $\sigma_i$ to mean the $2\times 2$ matrix  $\sigma$ acting on $i$ tensored with the $\mathbb{I}\in \mathbb{C}^{2\times 2}$ acting on each of the other qubits (the total number of qubits, $n$, will be clear from context when this notation is used).  The formal definition of Quantum Max Cut, $QMC(G, w)$ is:  
\begin{definition}[ $QMC(G, w)$ ]\label{def:qmc}
Given a graph $G=(V, E)$ with $|V|=n$, let $H\in \mathbb{C}^{2^n \times 2^n}$ such that:
\begin{align*}
H=\sum_{ij \in E} w_{ij}(\mathbb{I}- X_i X_j-Y_i Y_j-Z_i Z_j)
\end{align*}
Then, we define $QMC(G, w)$ to be the largest eigenvalue of $H$. Ideally, one also seeks to produce a (description) of a state achieving this value.
\end{definition}
\noindent Gharibian and Parekh~\cite{G19} introduced this problem as a maximization version of the well-known problem of finding ground states for the quantum Heisenberg model.  They give a classical $0.498$-approximation using product states, where a $\frac{1}{2}$-approximation is the best possible in the product state regime.  Anshu, Gosset, and Morenz~\cite{A20} present a classical rounding algorithm that outputs a description of an entangled state and are able to deliver a $0.531$-approximation.  To the best of our knowledge, this is the first approximation algorithm for a $2$-Local Hamiltonian problem to move beyond product states.  Likewise, the analysis in \cite{A20} differs from the analysis in the other related works.  Instead of using SDPs to upper bound the optimal quantum objective, \cite{A20} uses physical considerations for the particular kind of Hamiltonian they study \cite{L62}.  The key technical component is an upper bound on $QMC(G, w)$ where $G$ is a star graph. The rounding algorithm is also fundamentally different in that the output quantum state is produced from direct consideration of the Hamiltonian and its weights, rather than a solution to an SDP.  

Another important work for understanding our contribution is that of Brand${\rm\tilde{a}}$o and Harrow \cite{B16_v2}, since this paper makes use of higher levels of the quantum Lasserre Hierarchy.  Essentially the relevant rounding algorithm from this paper proceeds in the same way as the classical counterpart by Barak, Raghavendra, and Steurer~\cite{R12}, where a set of subsystems is sampled and all other density matrices are sampled according to single qubit density matrices conditioned on this set.  There are additional issues that arise in the analysis from the quantum-ness of the problem, but the rounding algorithm is semantically similar.  Additionally, all of the results presented in \cite{B16_v2} make strong non-local assumptions on the particular ``topology'' or structure of the instance.   

\textbf{Our contributions.}  In contrast to previous approaches, we make only local assumptions on the $2$-Local terms, and apply the second level of the Lasserre Hierarchy in a radical new way which makes crucial use of ``monogamy of entanglement'' inequalities.  Indeed, we believe that the methods we introduce constitute the most interesting contribution of this work.  

We bridge the gap between~\cite{A20} and more traditional SDP-based approximation algorithms by showing that the monogamy of entanglement bound derived in~\cite{A20}, based on a seminal result of Lieb and Mattis\cite{L62}, is a consequence of the second level of a quantum analog~\cite{D08, P10, B16_v2} of the classical Lasserre Hierarchy~\cite{L01, La01}.  To the best of our knowledge this is a first explicit example of such a connection.  This establishes the second level of the quantum Lasserre Hierarchy as the source of the best upper bound for Quantum Max Cut that is amenable to analysis. We show that weaker versions of this SDP relaxation, including the first level, fail to yield the monogamy of entanglement bound.
In addition we slightly improve upon the best-known approximation factor for $QMC$~\cite{A20} through a simple rounding algorithm that uses an SDP solution to guide construction of an entangled solution.  This is a significant departure from existing approximation algorithms for $2$-Local Hamiltonian problems, requiring new connections between quantum SDP relaxations and the convex hull of matchings in a graph.  Quantum Max Cut has emerged as a vehicle for advancement of approximation algorithms for $2$-Local Hamiltonian problems, since it maintains the hardness and essence of more general problems while hiding technical details that hinder progress~\cite{G19,A20,P20}.  We expect that the insights we develop here for Quantum Max Cut may be generalized for other problems.          


\textbf{Our methods.}  As stated previously, our rounding algorithm begins by formulating and solving an appropriate SDP, which comes from the quantum generalization of the Lasserre Hierarchy.  The SDP assigns a ``value'' for each edge, roughly corresponding to ``how close'' the parameters of the edge are to a singlet.  An edge with {\it large} value has parameters nearly matching the singlet.  Loosely speaking, if an edge has large value then the SDP ``thinks'' an optimal quantum solution is nearly a singlet along the edge.  The rounding algorithm proceeds by picking a threshold and adding every edge with value over the threshold to the large edge set (denoted $L$ in the paper).  In a legitimate quantum state, the concept of monogamy of entanglement implies that we cannot have too many large edges attached to the same vertex.  Since the SDP relaxation we use is relatively strong (\Cref{sec:l2_is_phys}), this implies the graph induced by the small edges must have low degree.  Hence, if we find a maximum matching on this graph, and place a singlet (the state in \Cref{eq:alg_output}) on each edge in the matching, we obtain a quantum state with performance approximately comparable to the SDP on this subgraph.  For the remainder of the qubits we place the maximally mixed state.  

Intuitively, this technique of thresholding the edges and then finding a matching has poor performance when all the edges have small values.  However, in this case a product state gives a good approximation to the objective: if all the edges are small then the state does not align well with the singlet along the edges in the Hamiltonian, hence entanglement is not really needed to emulate the state.  The rounding algorithm checks the value of both of these strategies (singlets on large edges vs. product state rounding) and takes whichever is better.  

\textbf{Future work.} Our analysis is not optimal, and it is possible to obtain improvements.  For example, we may consider stronger valid inequalities for Quantum Max Cut solutions arising from our relaxation, which we know exist through numerical experiments.  How far can such an improvement be pushed?  Can we significantly improve the approximation ratio for Quantum Max Cut beyond $\approx 0.53$?  Our analysis shows that the second level of the quantum Lasserre Hierarchy is exact for star graphs.  Can similar results be achieved for more interesting classes of graphs? 

Another important direction is the search for upper bounds on achievable approximation factors.  For classical optimization problems there are many such bounds known \cite{Kh07, Kh05}.  Most of these rely on a complexity theoretic conjecture referred to as the Unique Games Conjecture (UGC) \cite{Kh02}, i.e.\ if UGC holds then we have the corresponding upper bound on the approximation factor.  No analogous results are known for quantum optimization problems.

\section{Preliminaries}
We use standard quantum information and graph theory notation, highlighting a few specific definitions below.  

For an integer $l \geq 1$, we let $[l] := \{1,\ldots,l\}$.  For a set $S$, $\mathbb{R}^S$ refers to $\mathbb{R}^{|S|}$, where the dimensions of the Euclidean space are associated with the elements of $S$.  We generally refer to the elements of a vector $x \in \mathbb{R}^S$ as $x_l$ for $l \in S$; however, we will also refer to $x_l$ as variables comprising a solution $x$ in the context of semidefinite and linear programs. 

\textbf{Quantum information.}
The Pauli matrices take their usual definition:
\begin{equation*}
\label{eq:paulis}
 \mathbb{I}=\begin{bmatrix}
1 & 0 \\
0 & 1
\end{bmatrix},
\,\,\,\,\,\,
\X=\begin{bmatrix}
0 & 1 \\
1 & 0
\end{bmatrix},
\,\,\,\,\,\,
\Y=\begin{bmatrix}
0 & -i \\
i & 0
\end{bmatrix}, \,\text{and}
\,\,\,\,\,\,
\Z=\begin{bmatrix}
1 & 0 \\
0 & -1
\end{bmatrix}.
\end{equation*}
\noindent We follow the standard practice of using subscripts to indicate quantum subsystems among $n$ qubits, and we use the notation $\sigma_i$ to denote a Pauli matrix $\sigma \in \{X,Y,Z\}$ acting on qubit $i$, i.e.\ $\sigma_i := \mathbb{I} \otimes \mathbb{I} \otimes \ldots \otimes \sigma \otimes \ldots \otimes \mathbb{I} \in \mathbb{C}^{2^n \times 2^n}$, where the $\sigma$ occurs at position $i$. The sets $\mathcal{S}(\mathcal{X})$ and $\mathcal{H}(\mathcal{X})$ refer to the symmetric and Hermitian matrices, respectively, acting on the (complex) Euclidean space $\mathcal{X}$.


\textbf{Graph theory.} We deal with only finite and simple graphs $G=(V,E)$, with vertex set $V$ and edge set $E$.  The notation $E(G)$ is the edge set of a graph $G$. We will refer to an edge $e$ with endpoints $i,j \in V$ as $ij \in E$, or simply as $e \in E$ when endpoints are immaterial.  We generally consider weighted graphs where a weight $w_e \geq 0$ is specified for each edge $e \in E$. 

For a graph $G=(V,E)$, and a set of vertices $S \subseteq V$, we denote the \emph{induced subgraph} on $S$, consisting of all edges in $E$ with both endpoints in $S$, as $G[S]$.  For a set of vertices and edges $S \subseteq V$ and $F \subseteq E$, respectively, the edge set $\delta_F(S)$ is defined as $\{ij \in F \mid |\{i,j\} \cap S| = 1\}$, and $E_F(S) := \{ij \in F \mid |\{i,j\} \cap S| = 2\}$.  We drop the subscript $F$ when $F=E$, and for a vertex $i \in V$, we abbreviate $\delta_F(\{i\})$ as $\delta_F(i)$. 

A graph is \emph{$k$-vertex connected} if it has at least $k$ vertices and deleting any set of fewer than $k$ vertices (and any incident edges) leaves a connected graph. A \emph{matching} $M$ is a set of edges such that no two distinct $e,f \in M$ share a common vertex. A \emph{perfect matching} in $G$ is a matching of size $\frac{|V|}{2}$.

\subsection{Approximation Algorithm Overview}
The formal rounding algorithm we propose is presented in Algorithm \ref{alg:main}.
\begin{figure*}
\begin{mdframed}
\begin{algorithm}[H]
\begin{enumerate}
    \item  Given as input a graph $G=(V, E)$ with weights $w=\{w_e \geq 0\}_{e \in E}$, solve Lasserre$_2(G,w)$ (\Cref{def:lasserre_k}).  Let the matrix $M$ be an optimal solution.  
    
    \item  \label{step:values}For each $ij\in E$ calculate $v_{ij}:=[M(X_i X_j, \mathbb{I})+M(Y_i Y_j, \mathbb{I})+M(Z_i Z_j, \mathbb{I})]/3$, where $M(\Gamma,\Phi)$ refers to the $(\Gamma,\Phi)$ entry of the matrix $M$.  Set $x_{ij}:=-v_{ij}$.
    
    \item \label{step:thresh} Pick an integer $d \geq 1$, and define $L := \{e \in E \mid x_e > \alpha(d) := \frac{d+3}{3(d+1)}\}$. Find a maximum-weight matching $F$ in the graph $G_L := (V,L)$ with respect to the weights $\{w_e\}_{e \in L}$.  Let $U$ be the vertices unmatched by $F$.
    
    \item  Define a quantum state:\footnotemark 
    \vspace{-8 pt}
    \begin{equation}\label{eq:alg_output}
        \rho_F:=\prod_{ij\in F} \left( \frac{\mathbb{I}-\X_i \X_j-\Y_i \Y_j-\Z_i \Z_j}{4}\right) \prod_{v\in U} \frac{\mathbb{I}_v}{2}.
    \end{equation}
    \vspace{-8 pt}
    \item \label{step:prod-state} Execute the randomized approximation algorithm for Quantum Max Cut from~\cite{G19}, yielding a product state $\rho_{PS}$ from a Lasserre$_1$ solution.
    
    \item  Output the better of $\rho_F$ and $\rho_{PS}$.  
\end{enumerate}
\caption{\label{alg:rounding} Approximation Algorithm for Quantum Max Cut}\label{alg:main}
\end{algorithm}
\end{mdframed}
\end{figure*}\footnotetext{Recall $X_i$ is a tensor product of identity operators and a single $X$ operator in the $i$th position.  So, $((\mathbb{I}+X)/2)\otimes ((\mathbb{I}+X)/2)$ is expressed as $\prod_{i=1}^2 (\mathbb{I}+X_i)/2$ rather than $\bigotimes_{i=1}^2 (\mathbb{I}+X_i)/2$}
To understand the significance of the parameter $d$ in Step~\ref{step:thresh}, 
 recall that we find a set of ``large'' edges $L$ based on a threshold.  The strength of Lasserre$_2$ implies that $G_L$ has bounded degree.  $d$ is the degree upper bound we prove (\Cref{lem:degree-bound}) corresponding to threshold $\alpha(d)=(d+3)/(3(d+1))$.  In particular, if $d=1$  then no vertex has two adjacent edges and we may select all edges in $L$ for our matching.  The problem with this strategy, however, is that if all the edges have small values then the product state rounding algorithm (Step~\ref{step:prod-state}) has relatively poor performance.  Hence, we obtain the result for $d=2$.  This allows us to get better performance for product state rounding but requires more work to show a maximum matching has good performance with respect to the SDP.  \knote{Rewrote this paragraph.  It's repeating some stuff from the intro, but the ref wanted more exposition here}

\textbf{Analysis outline.}  The main theorem of this work (\Cref{thm:main}) proves the stated approximation factor of \Cref{alg:rounding}.  The proof of this theorem requires first demonstrating (in \Cref{sec:l2_is_phys}) several inequalities on the optimal solution of the second level of the quantum Lasserre Hierarhcy (demonted Lasserre$_2$).  Roughly there are two sets of techniques we use to prove the inequalities we need.  The first set (\Cref{sec:rel_stren_relax}) involves using invariance of the the objective function under certain permutations of the SDP variable and Schur complements.  The second set of bounds follows from sum-of-squares proof techniques \Cref{sec:SOS}.  

Understanding the performance of the thresholding (Step~\ref{step:thresh} in the algorithm) involves showing that constraints satisfied by the SDP (\Cref{sec:l2_is_phys}) imply that the ``large'' edges $L$ can be scaled by a not too small constant and brought into the convex hull of matchings (\Cref{thm:edmonds-matching}).  This provides a lower bound on the performance of the state $\rho_F$, then we may appeal to~\cite{G19} to lower bound the performance of $\rho_{PS}$.  We will prove the main theorem first, using components proved subsequently.  The reader is encouraged to come back to this proof after reading the document

\begin{theorem}[Main Theorem]\label{thm:main}
Let $G=(V, E)$ be a graph and $\{w_e\}_{e\in E}$ be a set of weights with $w_e \geq 0$ for all $e\in E$.  Let $H$ be the $QMC$ Hamiltonian in \Cref{def:qmc}, and let $\rho$ be the density matrix output by \Cref{alg:rounding}.  Then,
\begin{equation*}
    \frac{\mathbb{E}[\Tr (H\rho)]}{QMC(G, w)} \geq 0.533,
\end{equation*}
where the numerator is the expected objective value obtained by \Cref{alg:rounding}
\end{theorem}
\begin{proof}
Let $d=2$, let $\{x_e\}_{e\in E}$ be the values obtained from the SDP as in \Cref{step:values}, let $L$ be the set of edges found in \Cref{step:thresh}, let $S:=E-L$, and let $\{y_e^*\}_{e\in E}$ be such that $\{0, 1\} \ni y_e^*=1$ if and only if edge $e$ is chosen in the matching for $\rho_F$  (see \Cref{lem:matching_bound}).

Define:
\begin{equation*}
    s:=\frac{\sum_{e \in S} w_e (1+3 x_e)}{\sum_{e \in S} w_e (1+3 x_e)+\sum_{e \in L} w_e(1+3 x_e)},
\end{equation*}
\noindent and note that $s\in [0, 1]$ since the comment below \Cref{lem:L2-edge-bound} implies that $(1+3 x_e) \geq 0$ for $e \in E$. It holds that 
\begin{equation*}
    \frac{\sum_{e \in S} w_e (1+3 y_e^*)+\sum_{e \in L} w_e(1+3 y_e^*)}{\sum_{e \in S} w_e (1+3 x_e)+\sum_{e \in L} w_e(1+3 x_e)}
    =\frac{\sum_{e \in S} w_e (1+3 y_e^*)}{\sum_{e \in S} w_e (1+3 x_e)} s
    + \frac{\sum_{e \in L} w_e (1+3 y_e^*)}{\sum_{e \in L} w_e (1+3 x_e)} (1-s)
\end{equation*}
\noindent Now we can apply \Cref{lem:matching_bound},
\begin{equation*}
    \frac{\sum_{e \in S} w_e (1+3 y_e^*)}{\sum_{e \in S} w_e (1+3 x_e)} s
    + \frac{\sum_{e \in L} w_e (1+3 y_e^*)}{\sum_{e \in L} w_e (1+3 x_e)} (1-s)\geq \frac{3}{8} s +\frac{3}{4}(1-s).
\end{equation*}
A similar argument for $\rho_{PS}$ using \Cref{lem:prod_bound} yields:
\begin{equation*}
    \frac{\mathbb{E}[\Tr(H \rho_{PS})]}{\sum_{e \in S} w_e (1+3 x_e)+\sum_{ij\in L} w_e(1+3 x_e)} \geq 0.557931 s + 0.498766 (1-s) 
\end{equation*}
A lower bound on the expected approximation factor is
\begin{equation*}
\min_{s \in [0,1]} \max \left\{\frac{3}{8} s +\frac{3}{4}(1-s),\ 0.557931 s + 0.498766 (1-s)\right\},
\end{equation*}
which is calculated by the linear program,
\begin{align*}
    0.533 \leq& \min r\\
    s.t.\ \  &\left\{\frac{3}{8} s +\frac{3}{4}(1-s) \leq r,\ 
    0.557931 s + 0.498766 (1-s) \leq r,\ 
    1 \geq s \geq 0\right\}.
\end{align*}

\end{proof}

\section{The Level-2 Quantum Lasserre Hierarchy}\label{sec:l2_is_phys}

\subsection{Definitions}

The classical or commutative Lasserre Hierarchy (and the dual Sum-of-Squares Hierarchy) is a set of semidefinite programs which relaxes the notion of a probability distribution to a pseudo-distribution~\cite{B12}.  A pseudo-distribution is an assignment of values to low order moments which respects some, but not all, of the properties that a fully consistent probability distribution would satisfy. To understand this consider $n$ binary random variables $(A_1, ..., A_n)$.  We will be interested in expectations of polynomials in the $A_i$.\knote{ deleted:``, which are multilinear, without loss of generality, since the $A_i$ are binary'', it is pulling focus} For each monomial of degree $t \leq 2k$ in these variables, the level-$k$ instance of the hierarchy assigns value: $v_k(A_{i_1} A_{i_2}\ldots A_{i_t}) \in [0, 1]$.  The value $v_k$ is meant to represent the expectation $\mathbb{E}_{\mathcal{D}}[A_{i_1} A_{i_2}\ldots A_{i_t}]$ for a valid probability distribution $\mathcal{D}$, but it is also possible that it assigns values in such a way that it is impossible to have $v_k(A_{i_1} A_{i_2}\ldots A_{i_t})=\mathbb{E}_{\mathcal{D}}[A_{i_1} A_{i_2}\ldots A_{i_t}]$ for any valid distribution $\mathcal{D}$.  The level-$k$ SDP assigns values so that polynomials of degree at most $k$ behave as they should for a valid distribution.  In particular the SDP assigns values to monomials in such a way that if one expanded $p(A_1,\ldots, A_n)^2$ as a linear combination of monomials and applied $v_k$ to the individual terms, the resulting value $v_k(p(A_1,\ldots, A_n)^2)\geq 0$. Note that the expected behavior for random variables is the same: $\mathbb{E}_{\mathcal{D}}[p(A_1,\ldots, A_n)^2]\geq 0$ for a distribution $\mathcal{D}$.  The level $k$ can be thought of as checking that the distribution looks valid from the perspective of low order polynomials.

A quantum analog of the Lasserre Hierarchy~\cite{D08, P10, B16_v2} is essentially the same except that it is checking the validity of low order polynomials in the Pauli matrices with respect to an overall quantum distribution (density matrix).  The values we will assign are meant to represent values of $\Tr(\Gamma \rho)$ for $\Gamma$ a ``low-order'' tensor product of Pauli matrices and $\rho $ a valid density matrix.  However, the relaxation will likely assign values $v(\Gamma)$ in such a way that it is impossible for $v(\Gamma)=\Tr(\Gamma \rho)$ to hold for any density matrix (and for all $\Gamma$)\footnote{Indeed, if we were able to constrain the low order statistics to be {\it globally consistent} with some (physical) density matrix, then we could find the largest eigenvalue and solve a QMA-complete problem\cite{L06}.  }.  In this context, by ``low-order monomial'' we mean the following:
\begin{definition}[$\mathcal{P}_n(k)$]
Given $k$, $n$ define $\mathcal{P}_n(k)$ as the set of Pauli operators of weight $\leq k$.  Formally, $\Gamma\in \mathcal{P}_n(k)$ if $\Gamma$ is a tensor product of $n$ operators, each of which is in $\{\I, \X, \Y, \Z\}$ such that at most $k$ are not $\I$.
\end{definition}

Lasserre$_k$ will assign values to monomials (elements of $\mathcal{P}_n(2k)$) in such a way that if $p=\sum_{\Phi \in \mathcal{P}_n(k)} c_\Phi \Phi$, then $v(p^2)=\sum_{\Phi, \Phi'} c_\Phi c_{\Phi'} v(\Phi \Phi')\geq 0$.  A value assignment which respects low order statistics is equivalent to a positive-semidefinite (PSD) constraint on a ``moment matrix''.  To understand this imagine we had a PSD matrix $M$ with rows and columns indexed by elements of $\mathcal{P}_n(k)$, and we assigned values so that $v(\Gamma):=M(\Phi, \Psi)$ if $\Phi \Psi=\Gamma$.  Then, given some polynomial $p$, 
\begin{equation*}
    v(p^2)=v\left( \left(\sum_{\Phi\in \mathcal{P}_n(k)}  c_{\Phi} \Phi\right)^2\right)=\sum_{\substack{\Phi, \Phi'\\ \in \mathcal{P}_n(k)}} c_\Phi c_{\Phi'} v(\Phi\Phi')= \sum_{\substack{\Phi, \Phi'\\ \in \mathcal{P}_n(k)}} c_\Phi c_{\Phi'} M(\Phi, \Phi')=c^T M c,
\end{equation*}
where $c \in \mathbb{R}^{\mathcal{P}_n(k)}$ is the vector of monomial coefficients.  Since $M$  was assumed PSD we are guaranteed that the RHS is $\geq 0$.  Indeed, if we assign values based on a PSD matrix subject to appropriate constraints, we are guaranteed that Lasserre$_k$ will respect low degree polynomials:  
\begin{definition}[Lasserre$_k(G, w$)]\label{def:lasserre_k}
Given $k$, a graph $G=(V, E)$ on $n$ vertices, as well as a vector of non-negative weights $\{w_e\}_{e\in E}$ let $\gamma:=|\mathcal{P}_n(k)|$.  For each $ij\in E$, define $C_{ij}\in \mathcal{S}(\mathbb{R}^{\gamma \times \gamma})$ where rows and columns are indexed by elements of $\mathcal{P}_n(k)$ such that
\begin{alignat*}{2}
    C_{ij}(\sigma_i, \sigma_j) &:=-\frac{1}{2} && \quad \text{for $\sigma\in \{\X, \Y, \Z\}$},\\
    C_{ij}&:=0 && \quad \text{otherwise}.
\end{alignat*}
Let $M\in \mathcal{S}(\mathbb{R}^{\gamma \times \gamma})$ be an SDP variable with rows and columns indexed by elements of $\mathcal{P}_n(k)$.  We define Lasserre$_k(G, w)$ as the following SDP:
\begin{alignat}{2}
    \max \sum_{ij\in E} & \mathrlap{w_{ij} \left(1+\Tr(C_{ij} M) \right)}\\
    s.t. \qquad 
    \label{eq:diag_constraint}M(\Gamma, \Gamma)&=1 \quad && \forall\, \Gamma \in \mathcal{P}_n(k),\\
    \label{eq:hermitian_constraint} M(\Gamma, \Phi)&=0 \quad && \forall \, \Gamma, \Phi \in \mathcal{P}_n(k) \text{  s.t. $\Gamma \Phi$ is not Hermitian}, \\
    \label{eq:composition_constraint}M(\Gamma, \Phi)&=M(\Gamma', \Phi') \quad && \forall \, \Gamma, \Phi, \Gamma',\Phi' \in \mathcal{P}_n(k)\text{\,\,s.t. $\Gamma \Phi=\Gamma' \Phi'$},\\
    \label{eq:neg_composition_constraint}M(\Gamma, \Phi)&=-M(\Gamma', \Phi') \quad && \forall \, \Gamma, \Phi, \Gamma',\Phi' \in \mathcal{P}_n(k) \text{ s.t. $\Gamma \Phi=-\Gamma' \Phi'$}, \\
    M &\semigeq 0,\\
    M &\in \mathcal{S}(\mathbb{R}^{\gamma \times \gamma}).
\end{alignat}
We will denote Lasserre$_k(G)$ as the above problem with uniform weights (set all $w_{ij}=1$).  Note that we employ a real version of the Lasserre Hierarchy rather than the usual complex version.  This still provides an upper bound on the optimal quantum state as shown below in \Cref{thm:is_relaxation}.
\end{definition}
\noindent Since in Lasserre$_2$ we have constraints $M(\sigma_i, \sigma_j)=M(\sigma_i \sigma_j, \mathbb{I})$, we could have equivalently defined the objective matrix using these moment matrix entries, i.e.\ taking $C_{ij}(\sigma_i \sigma_j, \mathbb{I})\neq 0$.  

\begin{definition}[Lasserre$_k$ Edge Values]\label{def:lasserre-edge-values}
From a solution $M$ to Lasserre$_k$(G,w) (\Cref{def:lasserre_k}), we define edge values that are used by the rounding algorithm, \Cref{alg:rounding}.  Such values are defined for every pair of distinct vertices $i,j$, hence we assume, when referring to these values, that $E$ is edge set of a complete graph, denoted $K_n$.  We may set $w_e = 0$ for edges $e \in E$ that do not contribute to the objective value.  We define:
\begin{equation*}
v_{ij} := \frac{M(X_i X_j, \mathbb{I}) + M(Y_i Y_j, \mathbb{I}) + M(Z_i Z_j, \mathbb{I})}{3} \text{, and } x_{ij}:=-v_{ij},
\end{equation*}
for all $ij \in E := E(K_n)$. We say an edge is \emph{large} if $x_{ij}\approx 1$ (and $v_{ij}\approx -1$).
\end{definition}

We will also need to define a modified version of Lasserre$_1$.  This is simply Lasserre$_1$ supplemented with positivity of $2$-qubit marginals.  This is a relaxation of intermediate strength between Lasserre$_1$ and Lasserre$_2$, so we have denoted is Lasserre$_{1.5}$.  The additional marginal constraints are crucial for the analysis presented in \cite{P20}, so a precise understanding of its strength is very interesting.  We will define it only for unweighted graphs, since it will not be used in the context of the approximation algorithm.  

\begin{problem}[Lasserre$_{1.5}(G)$]\label{prob:5}
Given a graph $G=(V, E)$ on $n$ vertices, for each $ij \in E$ let $C_{ij}$ be as defined in Lasserre$_1$.  Solve the following SDP:

\begin{alignat}{2}
\label{sdp-relax:obj} \max \sum_{ij\in E} & \mathrlap{\left(1+\Tr(C_{ij} M)\right)}\\
s.t.\qquad
 \label{sdp-relax:diag} M(\Gamma, \Gamma) &= 1 \quad && \forall \Gamma\in \mathcal{P}_n(1),\\
 \label{sdp-relax:diag-block-zero} M(\Gamma,\Phi) &= 0 \quad && \forall \Gamma, \Phi \in \mathcal{P}_n(1) \text{ s.t. $\Gamma\Phi$ is not Hermitian} ,\\
 \label{sdp-relax:2-marginals} M(\sigma_i,\eta_j) &= \Tr[\sigma \otimes \eta\ \rho_{ij}] \quad &&\forall ij \in E \text{ and $\sigma$, $\eta \in \{\X, \Y, \Z\}$},\\
 \label{sdp-relax:1-marginals-i} M(\sigma_i, \mathbb{I}) &= \Tr[\sigma_i \otimes \mathbb{I}\ \rho_{ij}] \quad &&\forall ij\in E \text{ and $\sigma\in \{\X, \Y, \Z\}$},\\
 \label{sdp-relax:1-marginals-j} M(\sigma_j, \mathbb{I}) &= \Tr[\mathbb{I} \otimes \sigma_j\ \rho_{ij}] \quad &&\forall ij\in E \text{ and $\sigma \in \{\X, \Y, \Z\}$},\\
 \label{sdp-relax:p_ij-trace} \Tr[\rho_{ij}] &= 1 \quad &&\forall ij \in E,\\
 \label{sdp-relax:p_ij-pos}  \rho_{ij} &\semigeq 0 \quad &&\forall ij \in E,\\
 \rho_{ij} &\in \mathcal{H}(\mathbb{C}^{4 \times 4}) \quad &&\forall ij \in E,\\
 \label{sdp-relax:M-pos} M &\semigeq 0,\\
 M &\in \mathcal{S}(\mathbb{R}^{(3n+1)\times(3n+1)}).
\end{alignat}
\end{problem}

Note that the main difference between Lasserre$_{1.5}$ and Lasserre$_1$ is the presence of constraints \Cref{sdp-relax:2-marginals}-\Cref{sdp-relax:1-marginals-j}.  As stated previously, their intent is to force consistency of $2$-local moment matrices by forcing them to correspond to physical $2$-qubit density matrices.  These relaxations are important because they relax quantum states, hence can be used as upper bounds on $2$-Local Hamiltonian problems:

\begin{theorem}\label{thm:is_relaxation}
For any constant $k$ Lasserre$_k$ is an efficiently computable semidefinite program that provides an upper bound on $QMC(G, w)$.
\end{theorem}

\begin{proof}
\knote{I changed this paragraph, please check}\onote{added the bit about the objective. This may be too detailed if we want to save space} Except for $M\semigeq 0$, the constraints and objective are affine on the entries of $M$, hence we do indeed have an SDP.  Since $M$ is of polynomial size (it has length $O(n^k)$ on one side), and there are polynomially many linear constraints ($O(n^{2k})$ many), the usual considerations show computational efficiency: All feasible $M$ have bounded norm since moment matrices are constrained to be $1$ along the diagonal, the identity matrix is feasible so strong duality holds, and there is a ``ball'' of operators around the identity which are feasible.  Hence, the program can be solved to arbitrary additive precision in polynomial time via the ellipsoid or interior point methods (e.g.,~\cite{V96}). 

Let $\ket{\psi}$ be an eigenvector corresponding to $\lambda_{max}(H)$ where $H$ is the $2$-Local Hamiltonian in $QMC$ (\Cref{def:qmc}).  Set $M(\Phi, \Gamma)=\Tr(\Phi \Gamma \rho)$.  $M$ is PSD since for a complex vector $v$, $v^\dagger M v=\Tr(S^2 \rho)$ for $S$ some polynomial as previously described.  The remaining issue is that if $\Phi\Gamma$ is not Hermitian then the corresponding value of $M$ is purely imaginary, so we may not be satisfying \Cref{eq:hermitian_constraint}.  The solution is simply to set a new moment matrix $M'$ as $M'=(M+M^*)/2$ where $M^*$ is the same as $M$ but with complex conjugate entries.  Note that $M'$ is PSD since $M^*$ must also be PSD.  For the objective, note that 
$$
\Tr\left(w_{ij}(\mathbb{I}-X_i X_j - Y_i Y_j-Z_i Z_j) \ket{\psi}\bra{\psi} \right)=w_{ij}(1+\Tr(M C_{ij})).
$$

Hence we have established that the optimal quantum state has the same energy as the objective for {\it some} feasible $M$.  It follows that the optimal $M$ has objective which upper bounds the optimal quantum solution.  

\end{proof}

\subsection{Relative Strength of Relaxations}\label{sec:rel_stren_relax}

An important contribution of this work is that the level-$2$ instances of the Lasserre Hierarchy satisfy important physical constraints which are not satisfied by the first level, even when the first level is further constrained with positive $2$-qubit marginals (Lasserre$_{1.5}$).  The physical property of interest can be thought of as a ``monogamy of entanglement'' with respect to specific partitions of the quantum state.  Let $G=(V, E)$ be a graph on $n+1$ vertices with vertex set $\{0, 1, ..., n\}$.  Further, let the edge set be $E=\{(0, 1), (0, 2), ..., (0, n)\}$.  This graph is easily visualized as $n$ ``leaves'' connected to a central vertex $0$.  The Hamiltonian for $QMC(G)$, $H$, can be thought of as ``testing'' entanglement along the edges since it is testing overlap with respect to a maximally entangled state.  If a state $\ket{\psi}$ had value $\braket{\psi|H|\psi}=4n$, then $\ket{\psi}$ would appear to be maximally entangled along all the edges in $E$.  Indeed, this is impossible, and the value of the maximum possible $\braket{\psi|H|\psi}$ (or the maximum eigenvalue of $H$) is an important result for the analysis presented in \cite{A20}:

\begin{theorem}[Star Bound \cite{A20, L62}]\label{thm:star_bound}
If $G$ is a star graph with $n$ leaves, $QMC(G)=2(n+1)$.
\end{theorem}
\noindent This was proven by Anshu, Gosset, and Morenz~\cite{A20} using a well-known monogamy of entanglement result for the Heisenberg model on complete bipartite interaction graphs by Lieb and Mattis~\cite{L62}.

The first observation we have is that the Lasserre$_1$ SDP violates this property in a maximal sense.  By this, we mean that the optimal solution has objective $4n$, rather than $2(n+1)$.  Using the informal language we used to describe Lasserre$_k$, if we only tracked $1$-local Pauli polynomials, we would think it is possible to have a state which is maximally entangled along many overlapping edges:
\begin{theorem}\label{thm:star_bound_l1}
For $G$ a star graph on $n$ vertices, Lasserre$_1(G)$ has optimal objective $4n$.
\end{theorem}

\begin{proof}

\begin{figure}
    \hrule
    
    \caption{Matrices Needed For Proof of \Cref{thm:star_bound_l1}}
    \begin{subfigure}[t]{0.45\textwidth}
        \centering
        \begin{align*}
M(\sigma_i, \eta_k)=\begin{cases}
1 \text{ if $i=k$ and $\sigma=\eta$}\\
-1 \text{ if $i=0$ and $\sigma=\eta$}\\
1 \text{ if $i\neq 0 \neq k$, and $\sigma=\eta$}\\
1 \text{ if $\sigma_i=\mathbb{I}=\sigma_k$ }\\
0 \text{ otherwise}
\end{cases}.
\end{align*}
\caption{Formal Description of Optimal Moment Matrix}\label{eq:moment_for_1}
    \end{subfigure}
    \hfill
    \begin{subfigure}[t]{0.45\textwidth}
        \centering
        \renewcommand{\kbldelim}{[}
\renewcommand{\kbrdelim}{]}
\begin{align*}
  M = \kbordermatrix{
    &              0        & 1 & 2 & \hdots & n & \mathbb{I} \\
    0          & \mathbb{I} & -\mathbb{I}& -\mathbb{I} & \hdots & -\mathbb{I} & \mathbf{0}\\
    1          &-\mathbb{I} & \mathbb{I} & \mathbb{I}  & \hdots & \mathbb{I}  & \mathbf{0}\\
    2          &-\mathbb{I} & \mathbb{I} & \mathbb{I}  & \hdots & \mathbb{I}  & \mathbf{0}\\
    \vdots     &\vdots      & \vdots     & \vdots      & \ddots & \vdots      & \vdots\\
    n          & -\mathbb{I}& \mathbb{I} & \mathbb{I} & \hdots & \mathbb{I}           & \mathbf{0}\\
    \mathbb{I} & \mathbf{0} & \mathbf{0} & \mathbf{0} &  \hdots & \mathbf{0}  & 1
  }.
\end{align*}
\caption{Informal Description of Optimal Moment Matrix}\label{eq:informalM}
    \end{subfigure}
    \begin{subfigure}[t]{0.45\textwidth}
    \centering
    \begin{equation*}
       \begin{bmatrix}
1 & -1 & -1 & \hdots & -1\\
-1 & 1 & 1 & \hdots & 1\\
-1 & 1 & 1 & \hdots & 1\\
\vdots & \vdots & & \ddots & \vdots \\
-1 & 1 & 1 & \hdots & 1
\end{bmatrix}.
\end{equation*}
    \caption{One of the Diagonal Blocks of $M$}\label{eq:one_of_diag_blocks}
    \end{subfigure}
    \begin{subfigure}[t]{0.55\textwidth}
    \centering
    $$
\begin{bmatrix}  
1 & 1 \hdots & 1\\
\vdots & \vdots  & \vdots\\
1 & 1 \hdots & 1
\end{bmatrix}-\begin{bmatrix} -1 \\ -1 \\ \vdots \\-1 \end{bmatrix}\begin{bmatrix} -1 & -1 &  \hdots & -1 \end{bmatrix}\semigeq 0,
$$
    \caption{Schur Complement to complete the proof.}\label{eq:l1_schur}
    \end{subfigure}
\hrule
\end{figure}
We can demonstrate an optimal solution to Lasserre$_1(G)$ by showing that the solution which ``picks up'' all the edges is feasible.  Since the moment matrix must be PSD, and since the diagonal blocks are forced to be $\mathbb{I}$,  this is the maximum possible solution for any SDP variable and a solution with this objective value must be optimal.  Define moment matrix $M$ as in \Cref{eq:moment_for_1}.  We can describe this matrix pictorially by partitioning the rows and columns into blocks corresponding to individual qubits, i.e.\ block $i$ would correspond to indices $\{\X_i, \Y_i,\Z_i\}$, as well as a block corresponding to the $\mathbb{I}$ index.  With these partitions, we can write $M$ as in  \Cref{eq:informalM} where $\mathbf{0}$ denotes the zero vector $[0, 0, 0]$.  It is easy to see that $M$ satisfies all the linear constraints on the matrix elements, the remaining task is to prove it is PSD.  We may reshuffle the rows/columns to write $M$ in block diagonal form as four blocks, where one of the blocks is $1$ and three of the blocks have the form \Cref{eq:one_of_diag_blocks}.  Using the method of Schur complements\cite{Z06}, PSDness of this matrix above is equivalent to \Cref{eq:l1_schur}, which holds trivially.  Hence $M$ is block diagonal with PSD blocks so it must be PSD.  

\end{proof}

One direction for fixing this problem is to note that many of the low-order statistics present in the optimal moment matrix (see \Cref{eq:informalM}) are non-sensical even for very small states.  The submatrix corresponding to qubits $1$ and $2$, for instance, has the form:
\begin{align*}
\kbordermatrix{
    &              \X_1       & \Y_1 & \Z_1 & \X_2 & \Y_2 & \Z_2 \\
    \X_1          & 1 & 0 & 0 & 1 & 0 & 0\\
    \Y_1          &0 & 1 & 0  & 0 & 1  & 0\\
    \Z_1          &0 & 0 & 1  & 0 & 0  & 1\\
    \X_2     &1      & 0     & 0      &1 & 0      & 0\\
    \Y_2          & 0& 1 & 0 &0 & 1     &0\\
    \Z_2 & 0&0 & 1 &  0 & 0  & 1
  }
\end{align*}
which is impossible for any (reduced) two qubit density martrix $\rho_{12}$ \cite{H96}.  This can be remedied in the SDP by adding variables $\rho_{ij}\in \mathbb{C}^{4\times 4}$ for every pair of vertices $i, j$ which force submatrices of the above form to conform to moment matrices of legitimate $2$-qubit quantum states, as in \cite{P20} and as in in Lasserre$_{1.5}$.

Unfortunately, this relaxation is still not strong enough to enforce the star bound (\Cref{thm:star_bound}):
\begin{theorem}\label{thm:stren_of_1point5}
If $G$ is a star graph on $n$ vertices, then Lasserre$_{1.5}(G)$ has optimal objective:
$$
n+3\sqrt{n(1+(n-1)/3)}.
$$
\end{theorem}

\begin{proof}
\begin{figure}
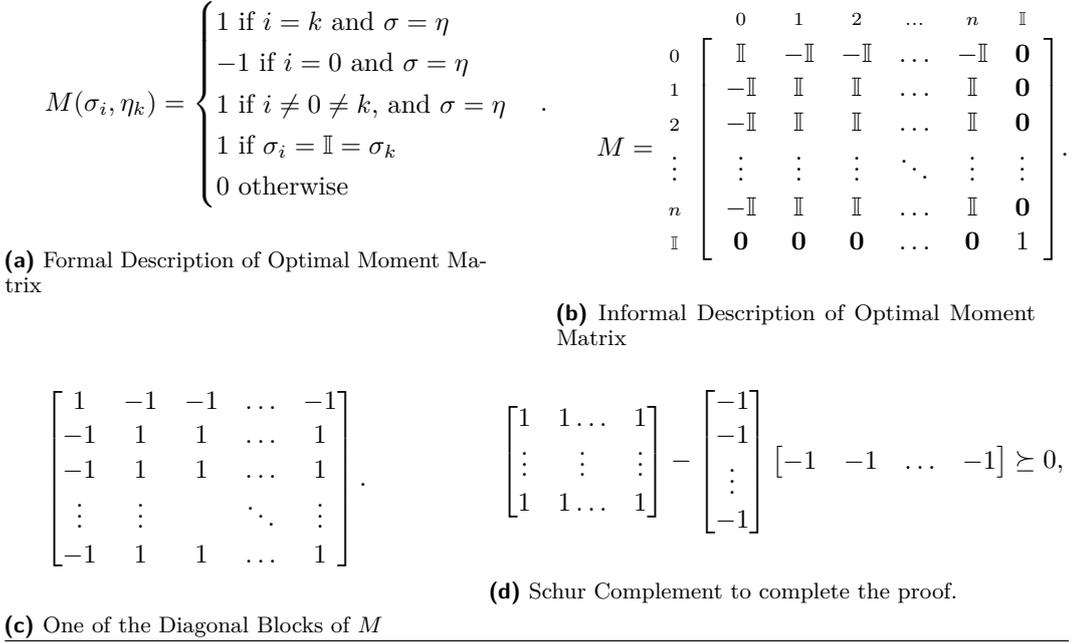

    \hrule
    
    \caption{Matrices Needed For Proof of \Cref{thm:stren_of_1point5}}
    \begin{subfigure}[t]{0.45\textwidth}
        \centering
        \begin{align*}
M = \kbordermatrix{
    &              0        & 1 & 2 & \hdots & n & \mathbb{I} \\
    0          & \mathbb{I} & P_1 & P_1 & \hdots & P_1 & \mathbf{0}\\
    1          & P_1^T & \mathbb{I} & P_2  & \hdots & P_2  & \mathbf{0}\\
    2          &P_1^T & P_2 & \mathbb{I}  & \hdots & P_2  & \mathbf{0}\\
    \vdots     &\vdots      & \vdots     & \vdots      & \ddots & \vdots      & \vdots\\
    n          & P_1^T & P_2 & \hdots & P_2 & \mathbb{I}           & \mathbf{0}\\
    \mathbb{I} & \mathbf{0} & \mathbf{0} & \mathbf{0} &  \hdots & \mathbf{0}  & 1
  },
\end{align*}
\caption{Symmeterized Optimal Moment Matrix}\label{eq:sym_moment_matrix}
    \end{subfigure}\label{eq:moment_for_1point5}
    \hfill
    \begin{subfigure}[t]{0.45\textwidth}
        \centering
        \begin{align*}
    Q:=\begin{bmatrix}
     1 & \alpha & \alpha & \hdots & \alpha\\
     \alpha & 1 & \beta & \hdots & \beta\\
     \alpha & \beta & 1 & & \beta \\
     \vdots & \vdots & & \ddots & \vdots\\
     \alpha & \beta & \hdots & \beta & 1
    \end{bmatrix},
\end{align*}

\caption{One Block of $M$}\label{eq:l15_one_block}
    \end{subfigure}
    \begin{subfigure}[t]{0.55\textwidth}
    \centering
    \begin{equation*}
     \begin{bmatrix}
      1 & \beta & \hdots & \beta\\
      \beta & 1 & & \beta \\
      \vdots & & \ddots & \vdots\\
      \beta & \hdots & \beta & 1
    \end{bmatrix}-
    \begin{bmatrix}
    \alpha \\ \alpha \\ \alpha \\ \vdots \\ \alpha 
    \end{bmatrix}\begin{bmatrix}
    \alpha & \alpha & \alpha & \hdots & \alpha 
    \end{bmatrix} \semigeq 0
\end{equation*}

    \caption{Schur Complement for Proof}\label{eq:l15_schur}
    \end{subfigure}
    \,\,\,\,
    \begin{subfigure}[t]{0.42\textwidth}
    \centering
    $$
U(\sigma_j, \sigma_i)=\begin{cases} 1 \text{ if $f(i)=j$ or $\sigma_i=\mathbb{I}=\sigma_j$}\\ 0 \text{ otherwise }\end{cases}.
$$
    \caption{Permutation Matrix for Symmeterization}\label{eq:l15_perm}
    \end{subfigure}
\hrule
\end{figure}

Let $M$ be the optimal moment matrix.  We may permute the leaves without changing the objective since the graph is unweighted, which corresponds to $M\rightarrow U M U^T$ for $U$ some orthogonal matrix.  Formally, if $f:[n]\rightarrow [n]$ is some permutation on the leaves, then define $U$ as in \Cref{eq:l15_perm} for all $\sigma\in \{\X, \Y, \Z\}$.  Then, the set of marginal density matrix variables can be redefined using $U M U^T$.  Similarly, we can permute $\{\X, \Y, \Z\}$ within each block.  We will assume that each $\{X_i, Y_i, Z_i\}$ is permuted in the same way so as to not change the objective.  By taking convex combinations of the described permutations (setting $M$ to be a convex combination of terms of the form $U M U^T$ for $U$ orthogonal matrices), we can then assume WLOG that $M$ has been fully symeterized and hence it is invariant to such permutations.  We can also assume that the entries $M(\sigma_i, \mathbb{I})$ are zero for $\sigma\in \{X, Y, Z\}$.  To see this observe that these terms do not participate in the objective, and that if they are non-zero in the optimal $M$ they can be set to zero without altering positivity.  We have shown that $M$ has the form \Cref{eq:sym_moment_matrix} where the rows/columns are indexed according to $\{X_0, Y_0, Z_0, X_1, Y_1, Z_1, ..., X_n, Y_n, Z_n, \mathbb{I}\}$.  Further, according to the symmeterization argument we can assume $P_1$ has the same entry on each diagonal, $P_2$ has the same entry on each diagonal and $P_2$ is symmetric.  Consider the submatrix corresponding to just the $\X$ operators, and denote if $Q$ as in \Cref{eq:l15_one_block} where $\alpha=M(\X_0,  X_i)$ and $\beta=M(X_1, X_2)$.  We know this matrix is PSD because it is a submatrix of a PSD matrix.  Further, since the submatrix corresponding to qubits $1$ and $2$ must correspond to the (valid) density matrix $\rho_{12}$, by \cite{H96, P20} $-1 \leq \beta \leq 1/3$.  Writing $Q$ in block diagonal form $\begin{bmatrix} A & B \\ B^T & C\end{bmatrix}$ with $A$ equal to the top left entry, and we can apply the method of Schur complements \cite{Z06} to obtain \Cref{eq:l15_schur}.

We can take the inner product with the all ones vector to obtain the inequality $n(1+(n-1)\beta)-\alpha^2 n^2\geq 0$.  Observe that this is a concave (down) parabola, to so to get a uniform bound on $\alpha$ we must set $\beta$ to the largest possible value.  Hence we derive $\alpha \geq -\sqrt{(1+(n-1)/3)/n}$ with $\beta=1/3$.  The objective is $n+(3 n \alpha)$, so far we have shown the optimal solution is at most $n+3\sqrt{n(1+(n-1)/3)}$.  To see the SDP achieves this upper bound note we can take $P_1\propto \alpha \mathbb{I}$ and $P_2\propto \beta \mathbb{I}$ for $\alpha $ and $\beta$ saturating the upper bound.  Just as in \Cref{thm:star_bound_l1}, $M$ decomposes into a block diagonal form where every block is PSD.  

\end{proof}

Lasserre$_2$, on the other hand, does satisfy the star bound, which will be an important fact for the analysis of our rounding algorithm. 

\begin{theorem}\label{thm:star-bound_l2}
Let $G=(V, E)$ be a star graph with $n$ leaves.  If $G$ is a subgraph of a larger graph $G'$ on $n'$ vertices, and if $M$ is any feasible solution to Lasserre$_2(G')$, then the values $\{x_e\}_{e \in E}$ (as in \Cref{alg:rounding}) satisfy $\sum_{e\in E} (1+3 x_e) \leq 2(n+1)$.  Additionally, Lasserre$_2(G)$ has optimal solution with objective $2(n+1)$.
\end{theorem}

\begin{proof}
This proof is very similar to \Cref{thm:stren_of_1point5}.  Let $M$ be the optimal solution to Lasserre$_2(G')$.  We will once again permute the leaves.  Consider some permutation of the leaves of $G$ $f:[n] \rightarrow [n]$.  Let $U\in \mathbb{C}^{2^{n'} \times 2^{n'}}$ be the (Clifford) unitary which satisfies $U \sigma_i U^\dagger=\sigma_j$ if $\sigma\in \{\X, \Y, \Z\}$ and $f(i)=j$.  Then, define the $|\mathcal{P}_{n'}(2)|\times |\mathcal{P}_{n'}(2)|$ permutation matrix $W$ as:
$$
W(\Gamma, \Phi)=\begin{cases} 1 \text{ if $U \Phi U^\dagger=\Gamma$}\\ 0 \text{ otherwise }\end{cases}
$$
\noindent and define $M'=W M W^T$.  Constraint \Cref{eq:composition_constraint} holds because if $\Gamma \Phi = \Theta \Pi$ then $(U^\dagger \Gamma U)(U^\dagger \Phi U)=(U^\dagger \Theta U) (U^\dagger \Pi U)$ so
\begin{align*}
M'(\Gamma, \Phi)=M(U^\dagger \Gamma U, U^\dagger \Phi U)=M(U^\dagger \Theta U, U^\dagger \Pi U)=M'(\Theta, \Pi)
\end{align*}
\noindent The other constraints are similar.  Just as in \Cref{thm:stren_of_1point5}, we can also permute $\{X, Y, Z\}$ in each block. Note that by permuting the leaves we have potentially reduced the objective value of $M$ overall (including the edges outside of the star graph), however we have not changed the sum of the values of edges in the star: $\sum_{e\in E} (1+3 x_e)$.  

Now extract the submatrix corresponding to the $2$-local terms along the edges corresponding to the $G$, as well as the identity.  Denote this as $Q$.  Formally, $Q$ is the submatrix induced by the index set 
\begin{align}\label{eq:index_set_star}
\{\X_0 \X_1, \Y_0 \Y_1, \Z_0 \Z_1, \X_0 \X_2, ..., \X_0 \X_n, \Y_0 \Y_n, \Z_0 \Z_n, \mathbb{I}\}.
\end{align}  
By symmeterization, we may assume $Q$ has the following block form\footnote{Note that off-diagonal elements of off-diagonal blocks are imaginary for a quantum state: $X_0 X_j \cdot{} Y_0 Y_k \propto i Z_0 X_j Y_k$.  So, those entries correspond to non-Hermitain matrices and are set to zero by constraint \Cref{eq:hermitian_constraint}}:
\begin{align}
\begin{tabular}{|c|c|c|c|c|}  \hline   $\begin{matrix}1 &-\alpha &-\alpha \\ -\alpha &1  &-\alpha \\-\alpha &-\alpha &1 \end{matrix}$ & 
$\begin{matrix}& & \\ &\beta \mathbb{I} & \\ & & \end{matrix}$ & 
$\begin{matrix}& & \\ &\hdots & \\ & & \end{matrix}$ & 
$\begin{matrix}& & \\ &\beta \mathbb{I} & \\ & & \end{matrix}$ & 
$\begin{matrix}\alpha \\ \alpha \\ \alpha \end{matrix}$\\ \hline
 $\begin{matrix}& & \\ &\beta \mathbb{I} & \\ & & \end{matrix}$ & 
 $\begin{matrix}1 &-\alpha &-\alpha \\ -\alpha &1  &-\alpha \\-\alpha &-\alpha &1 \end{matrix}$ & 
 $\begin{matrix}& & \\ &\ddots & \\ & & \end{matrix}$ & 
  $\begin{matrix}& & \\ &\vdots & \\ & & \end{matrix}$ & 
  $\begin{matrix}\alpha \\ \alpha \\ \alpha \end{matrix}$\\ \hline
$\begin{matrix}& & \\ &\vdots & \\ & & \end{matrix}$ &
   $\begin{matrix}& & \\ &\ddots & \\ & & \end{matrix}$ & 
   $\begin{matrix}& & \\ &\ddots & \\ & & \end{matrix}$ & 
   $\begin{matrix}& & \\ &\beta \mathbb{I} & \\ & & \end{matrix}$ & 
   $\begin{matrix}\\ \vdots \\ \end{matrix}$\\ \hline
   $\begin{matrix}& & \\ &\beta \mathbb{I} & \\ & & \end{matrix}$ & 
   $\begin{matrix}& & \\ &\hdots & \\ & & \end{matrix}$ & 
   $\begin{matrix}& & \\ &\beta \mathbb{I} & \\ & & \end{matrix}$ & 
   $\begin{matrix}1 &-\alpha &-\alpha \\ -\alpha &1  &-\alpha \\-\alpha &-\alpha &1 \end{matrix}$ & 
   $\begin{matrix}\alpha \\ \alpha \\ \alpha \end{matrix}$\\ \hline
    $\begin{matrix}\alpha & \alpha & \alpha \end{matrix}$ & 
    $\begin{matrix}\alpha & \alpha & \alpha \end{matrix}$ & 
    $\begin{matrix}\alpha & \alpha & \alpha \end{matrix}$ & 
    $\begin{matrix}\alpha & \alpha & \alpha \end{matrix}$ & 
    1\\ \hline \end{tabular}
\end{align}
\noindent Note that rows and columns are indexed as the same order as the set in \Cref{eq:index_set_star}.  By \Cref{lem:L2-edge-bound}, $-1 \leq \beta \leq 1/3$.  

Now observe that $Q\semigeq 0$ since it is a submatrix of a PSD matrix.  Consider this matrix in $2\times 2$ block form $\begin{bmatrix}A & B\\C & D \end{bmatrix}$ where $D=1$ is the bottom right entry.  Exactly as in previous proofs we can then use the method of Schur complement to derive a necessary condition for positivity: 
$$
-3 n \alpha^2-2 \alpha +(1+(n-1)\beta) \geq 0.
$$
Observe that this is a concave down parabola in $\alpha$, so to get a uniform bound we need to take $\beta$ as large as possible.  Setting $\beta=1/3$ and solving for the zeros of the polynomial yields the bounds $-(n+1)/(3n) \leq \alpha \leq 1/3$.  The lower bound yields the star bound for the sum of the edge values.  

To show Lasserre$_2(G)$ has optimal solution meeting the star bound note from the previous analysis we have that the optimal objective is upper bounded by the star bound.  Further, by \Cref{thm:star_bound} and the fact that we are considering a relaxation on the $2$-Local Hamiltonian problem (\Cref{thm:is_relaxation}) the optimal objective must be at least the star bound.  The result follows.
\end{proof}

\subsection{Valid Inequalities for Lasserre$_2$}\label{sec:SOS}
We now turn our attention to deriving inequalities that any Lassere$_2$ solution must satisfy.  These will be used in the subsequent analysis of \Cref{alg:rounding}.  Let $M$ be a feasible solution to Lasserre$_2(G,w)$.  Consider a Cholesky decomposition $U^T U=M \succeq 0$; we will refer to the vectors obtained from the columns of $U$ as $u(\Gamma) \in \mathbb{R}^l$, corresponding to $\Gamma \in \mathcal{P}_n(k)$ and the rows or columns of $M$.  Thus we have
\begin{equation}\label{eq:u-M-relationship}
u(\Gamma)^T u(\Phi) = M(\Gamma, \Phi) \quad \forall \Gamma,\Phi \in \mathcal{P}_n(k),
\end{equation}
and in particular each $u(\Gamma)$ is a unit vector by \Cref{eq:diag_constraint}. 
We will establish linear inequalities on the values $v_e$ of \Cref{def:lasserre-edge-values}, as these are the input to the rounding portion of \Cref{alg:rounding}.  First we establish bounds on the $v_e$.

\begin{lemma}\label{lem:L2-edge-bound}
For all $ij \in E$, $0 \leq 1 - M(X_i X_j, \mathbb{I}) - M(Y_i Y_j, \mathbb{I}) - M(Z_i Z_j, \mathbb{I}) \leq 4$.
\end{lemma}

\begin{proof}
We will use properties of the vectors $u(\Gamma)$ and $M$.  In particular suppose $\sigma, \eta, \tau \in \{X,Y,Z\}$ are distinct Paulis, giving us:
\begin{alignat*}{2}
u(\sigma_i \sigma_j)^T u(\eta_i \eta_j) &= M(\sigma_i\sigma_j,\eta_i \eta_j)&& \quad [\text{by \Cref{eq:u-M-relationship}}]\\ 
&= M(\sigma_i \eta_i \sigma_j \eta_j,\mathbb{I}) && \quad [\text{by \Cref{eq:composition_constraint}}]\\
&= -M(\tau_i\tau_j,\mathbb{I}) && \quad [\text{by \Cref{eq:neg_composition_constraint}}].
\end{alignat*}
The above in conjunction with $u(\Gamma)^T u(\Gamma)=1$ and $u(\sigma_i \sigma_j)^Tu(\mathbb{I}) = M(\sigma_i\sigma_j,\mathbb{I})$ yields the lower bound we seek to prove:
\begin{align*}
    0 &\leq [u(\mathbb{I}) - u(X_iX_j) - u(Y_iY_j) - u(Z_iZ_j)]^T [u(\mathbb{I}) - u(X_iX_j) - u(Y_iY_j) - u(Z_iZ_j)]\\
    &= 4[1 - M(X_i X_j, \mathbb{I}) - M(Y_i Y_j, \mathbb{I}) - M(Z_i Z_j, \mathbb{I})].
\end{align*}
For the upper bound, let the vector $z(1) := u(\mathbb{I}) - u(X_iX_j) + u(Y_iY_j) + u(Z_iZ_j)$. Analogously to above, we have
\begin{equation}
    \label{eq:bell-SoS}0 \leq \frac{1}{4}z(1)^T z(1) = 1 - M(X_i X_j, \mathbb{I}) + M(Y_i Y_j, \mathbb{I}) + M(Z_i Z_j, \mathbb{I}).
\end{equation}
If we let $z(2) := u(\mathbb{I}) + u(X_iX_j) - u(Y_iY_j) + u(Z_iZ_j)$ and $z(3) := u(\mathbb{I}) + u(X_iX_j) + u(Y_iY_j) - u(Z_iZ_j)$, then
\begin{equation*}
    0 \leq \sum_{l \in [3]} \frac{1}{4}z(l)^T z(l) = 3 + M(X_i X_j, \mathbb{I}) + M(Y_i Y_j, \mathbb{I}) + M(Z_i Z_j, \mathbb{I}),
\end{equation*}
by the analogs of \Cref{eq:bell-SoS} for $z(2)$ and $z(3)$. The above inequality is equivalent to the upper bound we seek to prove.
\end{proof}

The lemma implies that $-1 \leq v_e  \leq \frac{1}{3}$ and $-\frac{1}{3} \leq x_e \leq 1$, for all $e \in E$. We next derive inequalities for odd cycles in $G$.
\begin{lemma}\label{lem:odd-cycle-ineq}
For an odd-length cycle $C \subseteq E$, $\sum_{e \in C} v_e \geq 2-|C|$.
\end{lemma}

\begin{proof}
We take the same basic approach as the proof of \Cref{lem:L2-edge-bound}, namely the inequality will follow from a positive combination of inequalities derived from $z^Tz \geq 0$, for appropriately chosen vectors $z$.  First we consider the case when $|C|=3$; let $C$ consist of $ij,ik,jk \in E$, and let $\sigma \in \{X,Y,Z\}$.  Letting the vector $z(\sigma) := u(\mathbb{I}) + u(\sigma_i \sigma_j) + u(\sigma_i \sigma_k) + u(\sigma_j \sigma_k)$, we see:
\begin{equation*}
    0 \leq \frac{1}{4}z(\sigma)^T z(\sigma) = 1 + M(\sigma_i \sigma_j, \mathbb{I}) + M(\sigma_i \sigma_k, \mathbb{I}) + M(\sigma_j \sigma_k, \mathbb{I}),
\end{equation*}
since $u(\sigma_p \sigma_q)^T u(\sigma_q \sigma_r) = M(\sigma_p \sigma_q,\sigma_q \sigma_r) = M(\sigma_p\sigma_r,\mathbb{I})$, for $p,q,r \in V$.  Averaging the above inequality over $\sigma \in {X,Y,Z}$ yields:
\begin{equation}\label{lem:prf:triangle-ineq-pos}
    0 \leq \frac{1}{3}\sum_{\sigma \in \{X,Y,Z\}} \frac{1}{4}z(\sigma)^T z(\sigma) = 1 + v_{ij} + v_{ik} + v_{jk},
\end{equation}
establishing the lemma for $|C|=3$.  We establish another flavor of the triangle inequality above in order to extend the $|C|=3$ bound to larger cycles.  This time we let $z(\sigma) := u(\mathbb{I}) + u(\sigma_i \sigma_j) - u(\sigma_i \sigma_k) - u(\sigma_j \sigma_k)$, ultimately yielding:
\begin{equation}\label{lem:prf:triangle-ineq-neg}
    0 \leq \frac{1}{3}\sum_{\sigma \in \{X,Y,Z\}} \frac{1}{4}z(\sigma)^T z(\sigma) = 1 + v_{ij} - v_{ik} - v_{jk}.
\end{equation}
Let us derive the bound for $|C|=5$.  Suppose the vertices of $C$ are in $[5]$.  We sum three instances of the above inequalities: \Cref{lem:prf:triangle-ineq-pos} for triangles on $\{1,2,5\}$ and $\{2,3,4\}$, and $1 + v_{45} - v_{24} - v_{25} \geq 0$ for the triangle on $\{2,4,5\}$.  The sum is $3 + v_{12} + v_{23} + v_{34} + v_{45} + v_{15} \geq 0$, as desired.  More generally, for $|C|=2k+1$ with $k > 2$, we may sum $k$ instances of \Cref{lem:prf:triangle-ineq-pos} and $k-1$ instances of \Cref{lem:prf:triangle-ineq-neg} to derive the desired inequality.  These $2k-1$ triangles represent a triangulation of the cycle $C$; chords introduced by the triangulation appear in exactly two triangles, and edges of $C$ appear in exactly one triangle.  The $k-1$ instances of \Cref{lem:prf:triangle-ineq-neg} are used to cancel out variables on such chords.    
\end{proof}

The inequalities of the above lemma are actually implied by level $2$ of the classical Lasserre Hierarchy for the classical Max Cut problem.  This is captured by restricting $\mathcal{P}_n(2)$ to only contain tensor products of at most two Pauli Z's and setting $v_{ij} := M(Z_iZ_j,\mathbb{I})$.
\section{Analysis of Lasserre$_2$ Rounding}
Our goal is to provide bounds on the quality of the rounded solutions produced by \Cref{alg:rounding}, $\rho_F$ and $\rho_{PS}$, relative to our Lasserre$_2$ relaxation.  For each of these solutions, we consider both the contribution of the edges selected by \Cref{alg:rounding} to be in $L$ as well as those in $S := E-L$.

\subsection{Bounding the Quality of the Matching-Based Solution}
\Cref{alg:rounding} leverages a matching on a graph obtained by keeping only edges with high-magnitude fractional SDP values.  Here we show that the resulting graph has bounded degree and why this approach produces a matching of relatively large weight.  

We consider the values $x_e$ for $e \in E=E(K_n)$ (\Cref{def:lasserre-edge-values}) obtained from the Lasserre$_2(G,w)$ SDP (\Cref{def:lasserre_k}), so that the SDP's objective value is $1+3x_e$ for each $e \in E$. Recall from \Cref{alg:rounding} that we threshold these variables so that $L = \{e \in E \mid x_e > \alpha(d) \geq 0\}$, where $\alpha(d) = \frac{d+3}{3(d+1)}$.  The star bound allows us to bound the maximum degree in the resulting graph, $G_L = (V,L)$.
\begin{lemma}\label{lem:degree-bound}
The graph $G_L$, as defined above, has maximum degree at most $d$.
\end{lemma}

\begin{proof}
Suppose a vertex $i \in V$ in has degree at least $d+1$ in $G_L$.  Let $D \subseteq \delta(i)$ be a set of $d+1$ edges. We then have 
\begin{equation*}
\sum_{e \in D} (1+3x_e) > (d+1)(1+3\alpha(d)) = 2(d+2), 
\end{equation*}
violating \Cref{thm:star-bound_l2} for the star rooted at vertex $i$ and containing the edges in $D$. 
\end{proof}

\Cref{alg:rounding} finds a matching $F^*$ in $G_L$ that maximizes $\sum_{e \in F} w_e$ over matchings $F$ in $G_L$.  The algorithm obtains a quantum state $\rho_{F^*}$ from $F^*$ by putting a singlet, $\frac{1}{4}(\I-\X_i\X_j-\Y_i\Y_j-\Z_i\Z_j)$, on each edge in $F^*$ and a maximally mixed state, $\frac{1}{2}\I$, on each vertex unmatched by $F^*$.  Since the objective is $w_{ij}(\I-\X_i\X_j-\Y_i\Y_j-\Z_i\Z_j)$ for each edge $ij \in E$, this earns a weight of $4w_{ij}$ for every $ij \in F^*$ and a weight of $w_{ij}$ for every $ij \in E-F^*$.  If we define $y^*_e = 1$ for $e \in F^*$ and $y^*_e = 0$ otherwise, then we may express the weight earned by $\rho_{F^*}$ on an edge $e \in E$ as:
$4 w_e y^*_e + w_e(1-y^*_e) = w_e(1+3y^*_e)$.  We would like to show that the total weight of $F^*$ on the edges in $L$ is approximately the weight earned by the SDP on $L$: $\sum_{e \in L} w_e(1+3x_e)$.  The following lemma suggests a strategy for accomplishing this.

\onote{I rewrote the below, so please take a careful look.}
\begin{lemma}\label{lem:convex-combination-matchings}
If, for some $\beta \in [0,1]$, the vector $\beta x \in \mathbb{R}^{|E|}$ is a convex combination of matchings, then $\sum_{e \in L} w_e(1+3y^*_e) \geq \sum_{e \in L} w_e(1+3\beta x_e)$.
\end{lemma}
\begin{proof}
Write $\beta x = \sum_l \mu_l M_l$, where the latter is a convex combination of incidence vectors of matchings.  We have, for the vector of weights $w$, $\beta w^T x = \sum_l \mu_l w^T M_l$, so there is some $l'$ with $w^T M_{l'} \geq \beta w^T x$.  Since $y^*$ is the incidence vector of a maximum-weight matching with respect to $w$, we get $w^Ty^* \geq w^T M_{l'} \geq \beta w^T x$, completing our proof. 
\end{proof}

The hypothesis of the above lemma is equivalent to showing that $\beta x$ is in the convex hull of matchings for our graph $G$, and the convex hull of matchings in a graph is well understood.
\begin{theorem}[Pulleyblank and Edmonds~\cite{Pu74}; see~\cite{S03}, Section 25.2] \label{thm:edmonds-matching}
The convex hull of matchings in a graph $G=(V,E)$ is defined by the following linear inequalities:
\begin{alignat}{2}
 \label{matching:deg-constraints} \sum_{e \in \delta(i)} y_e &\leq 1 \quad && \forall i \in V,\\
  \label{matching:factor-critical-constraints} \sum_{e \in E(S)} y_e &\leq \frac{|S|-1}{2} \quad && \forall S \in \mathcal{F},\\
 \label{matching:nonneg-constraints} y_e &\geq 0 \quad && \forall e \in E,
\end{alignat}
where $\mathcal{F} := \{S \subseteq V \mid |S| \geq 3,\text{and }G[S]\text{ is factor critical and $2$-vertex connected}\}$. A graph $H$ is called factor critical (or hypomatchable) if deleting any vertex in $H$ leaves a graph containing a perfect matching (hence |H| must be odd). 
\end{theorem}
We obtain our main lemma by determining a relatively large $\beta \in [0,1]$ for which $y_e=\beta x_e$ is a feasible solution for the inequalities above, when we take $d=2$.
\begin{lemma}\label{lem:matching_bound}Suppose $v_e$ for $e \in E$ is a feasible solution to the SDP \Cref{def:lasserre_k}, and set $x_e := -v_e$. Let $L := \{e \in E \mid x_e > \frac{5}{9} = \alpha(2)\}$ and $G_L := (V, L)$ be the graph consisting of the edges in $L$.  If $F^*$ is an maximum-weight matching in $G_L$ with respect to the weights $w_e \geq 0$ for $e \in L$, then 
\begin{equation}\label{eq:lem:L-bound}
\frac{\sum_{e \in L} w_e(1+3y^*_e)}{\sum_{e \in L} w_e(1+3x_e)} > \frac{3}{4},
\end{equation}
where $\{0,1\} \ni y^*_e = 1$ if and only if $e \in F^*$.  If $S := E-L$, then
\begin{equation}\label{eq:lem:S-bound}
\frac{\sum_{e \in S} w_e(1+3y^*_e)}{\sum_{e \in S} w_e(1+3x_e)} \geq 
\frac{3}{8}.
\end{equation}
\end{lemma}

\begin{proof} We begin by considering \Cref{eq:lem:L-bound} and
first showing that the variables $\beta x_e$ satisfy the inequalities of \Cref{thm:edmonds-matching} for $G_L$ with $\beta = \frac{9}{14}$. Then, \Cref{lem:convex-combination-matchings} gives us
\begin{equation}\label{eq:lem:bound-simplification}
    \frac{\sum_{e \in L} w_e(1+3y^*_e)}{\sum_{e \in L} w_e(1+3x_e)} \geq
    \frac{\sum_{e \in L} w_e(1+3\beta x_e)}{\sum_{e \in L} w_e(1+3x_e)},
\end{equation}
and we may focus our attention on bounding the latter, which only depends on the $x_e$.  

\textbf{Satisfying the inequalities of \Cref{thm:edmonds-matching}.} Inequality~\eqref{matching:nonneg-constraints} is satisfied since $e \in L$ implies $x_e \geq 0$.  To see that the vector $\frac{9}{14}x$ is feasible for the inequality \eqref{matching:deg-constraints}, first note that since we take $d=2$, $G_L$ has maximum degree at most 2 by \Cref{lem:degree-bound}.  Inequality~\eqref{matching:deg-constraints} is satisfied for vertices of degree 1 since SDP \Cref{def:lasserre_k} gives us $x_e \leq 1$ for all $e$. Now another application of the star bound (~\Cref{thm:star-bound_l2}) to a degree-2 vertex, $i$ in $G_L$ with neighbors $j$ and $k$, gives us that 
\begin{equation}\label{eq:prf:degree-bound}
(1 + 3x_{ij}) + (1 + 3x_{ik}) \leq 6 \ \Rightarrow\ 
 x_{ij} + x_{ik} \leq \frac{4}{3},
\end{equation}
hence $\frac{9}{14}x_{ij} + \frac{9}{14}x_{ik} \leq \frac{3}{4}x_{ij} + \frac{3}{4}x_{ik} \leq 1$.

Next we will show that Inequality~\eqref{matching:factor-critical-constraints} is satisfied.  For these inequalities, we may assume that the induced subgraph on $S \in \mathcal{F}$ in $G_L$, $G_L[S]$, is an odd cycle.  The set $\mathcal{F}$ contains only odd-sized sets that are ($2$-vertex) connected, and $G_L$ has degree at most 2; hence, $G_L[S]$ must be a path or a cycle.  The graph $G_L[S]$ cannot be a path since it must be factor critical, and removing a penultimate vertex in a path leaves a graph with no perfect matching.

Pick some $S \in \mathcal{F}$, and sum the inequalities of \eqref{matching:deg-constraints} over $i \in S$. This yields
\begin{equation*}
    \sum_{e \in \delta_L(S)} y_e + \sum_{e \in E_L(S)} 2y_e \leq |S| \ \Rightarrow\ 
    \sum_{e \in E_L(S)} y_e \leq \frac{|S|}{2},
\end{equation*}
since $y_e \geq 0$ by Inequality~\eqref{matching:nonneg-constraints}.  This shows that any vector $y$ that satisfies Inequality~\eqref{matching:deg-constraints} gives a RHS of $\frac{|S|}{2}$ instead of the desired value, $\frac{|S|-1}{2}$ for Inequality~\eqref{matching:factor-critical-constraints}.  To make such a vector feasible for Inequality~\eqref{matching:factor-critical-constraints}, we must scale it by $\max_{k \geq 1} \frac{2k}{2k+1} = \frac{2}{3}$.  In our case, $\frac{3}{4}x$ satisfies Inequality~\eqref{matching:deg-constraints}, hence $\frac{2}{3}\cdot\frac{3}{4}x = \frac{1}{2}x$ is feasible for the inequalities of \Cref{thm:edmonds-matching}.  However, we can do better by considering additional inequalities satisfied by the $x_e$ that are implied by our Lasserre$_2$ SDP relaxation. \Cref{lem:odd-cycle-ineq} gives us that $x_{ij} + x_{ik} + x_{jk} \leq 1$ for $ij,ik,jk \in E$, hence the inequalities of \eqref{matching:factor-critical-constraints} for $|S|=3$ are satisfied by the vector $x$ (in fact, since $x_e > \frac{5}{9}$ for all $e \in L$, $G_L$ contains no triangles).  

For any cycle $C \subseteq L$ on $5$ vertices, \Cref{lem:odd-cycle-ineq} yields $\sum_{e \in C} x_e \leq 3$, so that $\sum_{e \in C}\frac{3}{4}x_e \leq \frac{9}{4}$.  Hence, $\frac{3}{4}x$ must be scaled by an additional factor of $\frac{8}{9}$ in order 
satisfy the inequalities of \eqref{matching:factor-critical-constraints} for $|S|=5$. For $|S| \geq 7$, an additional factor of $\max_{k \geq 3} \frac{2k}{2k+1} = \frac{6}{7}$ suffices.  Thus $\frac{6}{7}\cdot\frac{3}{4}x = \frac{9}{14}x$ is a feasible solution for the inequalities of $\Cref{thm:edmonds-matching}$, and it is consequently a convex combination of incidence vectors of matchings in $G_L$.

\textbf{Establishing \Cref{eq:lem:L-bound}.} For the RHS of \Cref{eq:lem:bound-simplification}, we have
\begin{equation}\label{eq:prf:L-edge-bound}
\frac{\sum_{e \in L} w_e(1+3\beta x_e)}{\sum_{e \in L} w_e(1+3x_e)} \geq
\min_{\{e \in L \mid w_e > 0\}} \frac{w_e(1+3\beta x_e)}{w_e(1+3x_e)},
\end{equation}
since $1+3x_e > 0$ for all $e \in L$, and the LHS above is a convex combination of the ratios $(1+3\beta x_e)/(1+3x_e)$ for $e \in L$ with $w_e > 0$. This reduces our task to bounding $(1+3\beta x_e)/(1+3x_e)$, for a worst-case value of $x_e$ achieving the minimum above.  If $e$ is an isolated edge in $G_L$ (i.e.\ it is not incident to any other edges in $L$), then we may assume $e$ is in the maximum-weight matching $F^*$ without loss of generality, since we can apply the arguments of this section to each connected component of $G_L$.  For such an edge $e$, we may simply take $\beta = 1$, yielding $(1+3\beta x_e)/(1+3x_e) =1$.  If $e$ is not isolated in $G_L$, then it is incident to another edge $f$ at a vertex of degree 2 in $G_L$.  Since $x_f > \frac{5}{9}$, by \Cref{eq:prf:degree-bound} we see that $x_e \leq \frac{4}{3} - x_f < \frac{7}{9}$. We consequently have, for $\beta = \frac{9}{14}$:
\begin{equation*}
    \min_{\{e \in L \mid w_e > 0\}} \frac{w_e(1+3\beta x_e)}{w_e(1+3x_e)} \geq \min_{x_e \in (\frac{5}{9},\,\frac{7}{9})} \frac{1+3\beta x_e}{1+3x_e} = 
    \beta + \min_{x_e \in (\frac{5}{9},\,\frac{7}{9})} \frac{1-\beta}{1+3x_e} > \frac{3}{10} + \frac{7}{10}\beta = \frac{3}{4},
\end{equation*}
demonstrating \Cref{eq:lem:L-bound}.

\textbf{Establishing \Cref{eq:lem:S-bound}.} We now turn our attention to the edges in $S = E-L$.  Since $F^*$ includes no edges in $S$, we have $y^*_e = 0$ for $e \in S$. By the definition of $S$, $x_e \leq \frac{5}{9}$ for $e \in S$. These facts yield \Cref{eq:lem:S-bound}.
\end{proof}

\textbf{Finding a maximum-weight matching in $G_L$.} We note that since each vertex in $G_L$ has degree at most 2 when $d=2$, each connected component of $G_L$ is a path or cycle.  In this case a maximum-weight matching may be found in linear time by a dynamic programming algorithm.

\subsection{Bounding the Quality of the Product State Solution}

We have established performance bounds on the matching part of the rounding algorithm.  The only remaining piece is a performance bound on the product state solution produced by the rounding algorithm, $\rho_{PS}$.
 
\begin{lemma}\label{lem:prod_bound}Suppose $v_e$ for $e \in E$ are values derived from the optimal solution to Lasserre$_2$, and set $x_e := -v_e$. Let $L := \{e \in E \mid x_e > \frac{5}{9} = \alpha(2)\}$ and let $S:=E-L$.  Then, with respect to the weights $w_e \geq 0$, the approximation algorithm from \cite{G19} produces a random product state $\rho$ satisfying:
\begin{equation}\label{eq:ps_L_bound}
\frac{\sum_{ij \in L} w_{ij}\mathbb{E}[\Tr((\mathbb{I}-\X_i \X_j-\Y_i \Y_j-\Z_i \Z_j)\rho)]}{\sum_{ij \in L} w_{ij}(1+3x_{ij})} \geq 0.498766,
\end{equation}
and
\begin{equation}
\frac{\sum_{ij \in S} w_{ij}\mathbb{E}[\Tr((\mathbb{I}-\X_i \X_j-\Y_i \Y_j-\Z_i \Z_j)\rho)]}{\sum_{ij \in S} w_{ij}(1+3x_{ij})} \geq 0.557931
.
\end{equation}
\end{lemma}

\begin{proof}
Let $M$ be the optimal solution to Lasserre$_2(G,w)$ produced by \Cref{alg:rounding}. The product state approximation algorithm of \cite{G19} relies on a feasible solution to Lasserre$_1(G,w)$, which we may obtain from $M$.  In particular the \cite{G19} algorithm takes as input the vectors $u(\Phi)$, from \Cref{eq:u-M-relationship}, for $\Phi \in \mathcal{P}_n(1)$ and rounds them to a product state solution.    

Let $\Gamma$ and $\,_2 F _1 $ be the Gamma and Hyergeometric functions as they are normally defined \cite{A48}.  The analysis of the \cite{G19} algorithm considers a worst-case edge $ij \in E$ and depends on the value $v_{ij} = -x_{ij}$.  The worst-case approximation ratio is determined by the quantity 
\begin{equation}\label{eq:GP-approx-ratio}
\min_{v_{ij} \in [-1,\frac{1}{3}]} \frac{1-F(v_{ij})}{1-3v_{ij}},
\text{ where }
F(t)=\frac{2}{3}\left( \frac{\Gamma(2)}{\Gamma(3/2)}\right)^2 \,_2 F_1 \left[ \begin{matrix}1/2 & 1/2\\ 5/2\end{matrix}; t^2\right].
\end{equation}
For more details see the paragraph above Section 4.1 in \cite{G19}, where $t $ in that paper is equal to $v_{ij}$ in our terminology. The first inequality, \Cref{eq:ps_L_bound} is immediate because it is the worst case approximation factor for their algorithm.  

The worst-case value of $v_{ij}$ in \Cref{eq:GP-approx-ratio} is close to $-1$. We take advantage of the fact that $v_{ij} \geq -\alpha(2)$ for $ij \in S$, avoiding the worst case. In particular we get a ratio of:    
\begin{equation*}
\min_{v_{ij} \in [-\frac{5}{9}, \frac{1}{3}]} \frac{1-F(v_{ij})}{1-3v_{ij}}= \frac{3}{8}\left(1-F(5/9)\right)\geq 0.557931.
\end{equation*}

\end{proof}

\bibliography{mybib}{}
\bibliographystyle{plainurl}

\end{document}